\RequirePackage{fix-cm}

\documentclass[smallextended]{svjour3}       

\smartqed  
\usepackage{graphicx}
 \journalname{arXiv}
\begin{document}

\title{Portfolio optimization in the case of an exponential utility function and in the presence of an illiquid asset
}

\titlerunning{Portfolio optimization in the case of an exponential utility function}        

\author{Ljudmila A. ~Bordag
}

\institute{Ljudmila A. ~Bordag\at
              Faculty of  Natural and Environmental Sciences, University of Applied Sciences Zittau/G\"orlitz,
Theodor-K\"orner-Allee 16,
D-02763 Zittau, Germany \\
              \email{ljudmila@bordag.com}
}

\date{Received: \today / Accepted: date}

\maketitle

\begin{abstract}
We study an optimization problem for a portfolio with a risk-free, a liquid, and an illiquid risky asset. The illiquid risky asset is sold in an exogenous random moment with a prescribed liquidation time distribution.  The investor prefers a negative or a positive exponential utility function.  We prove that both cases are connected by a one-to-one analytical substitution and are identical from the economic, analytical, or Lie algebraic points of view.

It is well known that the exponential utility function is connected with the HARA utility function through a limiting procedure if the parameter of the HARA utility function is going to infinity. We show that the optimization problem with the exponential utility function is not connected to the HARA case by the limiting procedure and we obtain essentially different results.

For the main three dimensional PDE with the exponential utility function we obtain the complete set of the nonequivalent Lie group invariant reductions to two dimensional PDEs according to an optimal system of subalgebras of the admitted Lie algebra. We prove that in just one case the invariant reduction is consistent with the boundary condition. This reduction represents a significant simplification of the original problem.

\end{abstract}

\section{Introduction}
\label{intro}
We study an optimization problem for a portfolio with an illiquid, a liquid risky, and a risk-free asset in the framework of continuous time. We suppose that the illiquid asset is sold in an exogenous random moment $T$ with a prescribed liquidation time distribution. When we started  investigating an optimization problem where the time-horizon is an exogenous random variable we have had in mind previous crises and some typical regulations in EU countries. But now the corona crises make this problem setting highly actual for most countries in Europe. It seems that in a few weeks the problem setting described in this paper will be very actual for many households, small and middle entrepreneurs as in summertime a lot of them will fear default. The previous crises showed that the situation will force them to sell their factory or house because they cannot pay obligations in time. It was a rather typical situation in previous crises. As it is set up now one cannot sell one piece at a time from own house or own small shop or factory; one cannot go with the price down, because it is just on the smallest possible level, and one needs a bit liquidity to stay alive and close most of own debts. During crisis time many investors change the risk tolerance dramatically and their behavior changes towards a strong conservative one. One can find a lot of references in the short review \cite{DiazEsparcia2019}  devoted to assessing risk tolerance in dependence on economic cycles as well as a discussion about different parameter choices and forms of utility functions and building financial models. We suppose that during an economic crisis some investors will prefer to use instead of a HARA utility function a CARA utility function, for instance an exponential utility function.

There is a huge amount of papers devoted to the classical optimal investment problem with a random endowment  as further generalizations of the famous Merton’s model like in \cite{Bouchard2004}, \cite{elkaroui:hal-00708493}. The main difference in our problem setting is that the portfolio includes an illiquid asset with a prescribed liquidation time distribution. For the first time the optimization problem in this form with the HARA utility function was introduced in the paper [4]  and later studied in papers  \cite {Bordag2016}, \cite {BordagYamshchikov2017},  \cite {BordagYamshchikovBook}.
If the illiquid asset in the portfolio is a real estate, a factory, a plant, or store then you can sell it as a whole only. It is a typical situation with the selling of houses, small factories or shops that you know the market situation and a typical liquidation time distribution. If it is the seller market then it will be similar to an exponential distribution. If it is the buyer market then the liquidation time distribution will be similar to a Weibull distribution. We study in this paper both liquidation time distributions as special cases.

The influence of the risk tolerance preferred by the investor on the solution of an optimization problem was studied before for different portfolio settings with and without an illiquid asset. We primary consider the results in connection with the selection of the utility function.

In the paper \cite{MoninZa2014} the authors started with the classical Merton's optimization problem used in \cite{Merton1969}. The portfolio contains one liquid risky asset and a risk-free money market account. The trading takes place within the fixed finite time horizon. The authors explore the question of risk management under different risk preferences of the investor. They study the optimal wealth process and the portfolio process across different utilities and provide transformations between two such processes corresponding to two arbitrary utilities. It is possible to find a deterministic transformation using the local absolute risk tolerance function associated with the corresponding utility function. This transformation is defined by a solution of a linear heat equation with the risk tolerance function as a coefficient by the second spatial derivative. Because of the classical problem setting it is possible to study the influence of the chosen utility and the risk tolerance on the wealth process and the different characteristics of the optimal portfolio in detail. The authors prove that the main role plays the curvature of the risk tolerance function of the preferred utility function. Certainly we cannot expect such tractability from an optimization  problem with an illiquid asset, but we can use this model as a benchmark for the case if the volume of the illiquid position of the studied portfolio vanishes.

The dependence of optimal liquidation strategies from the risk aversion of investors was studied in the work of Schied and Sch{\"o}neborn \cite{Schied2009}. The authors consider the infinite time horizon in the optimal portfolio liquidation problem and use a stochastic control approach. In this model a large investor trades one risky and one risk-free asset. Thereby due to insufficient liquidity of the risky asset the investor`s trading rate moves the market price for the risky asset.
The authors obtain for the value function and the optimal strategy nonlinear parabolic PDEs.  They have to determine the adaptive trading strategy that maximizes the expected utility of the proceeds of a large asset sale. Withal authors studied the financial influence of different types of investor's utility functions. They found that the optimal strategy is aggressive or passive in-the-money in dependence on the investor`s risk tolerance, i.e. if the utility function displays increasing or decreasing risk aversion. The authors proved that such strategies are rational for investors with different absolute risk aversion profiles.

Another approach to a liquidation problem of an illiquid asset is provided in paper \cite{Monin2014}. It is devoted to the problem of how efficiently liquidate large assets positions up to an exogenous fixed terminal time. The author supposes that the investor prefers the exponential utility function and seeks to maximize the expected utility of the terminal value of his wealth. The portfolio contains an illiquid asset called a primary risky asset, a liquid asset that is imperfectly correlated with the primary asset and is called a proxy risky asset as well as a riskless money market account that pays zero interest rate. In practice the investor tries to reduce the price impact by trading a large number of assets and to hedge market risk of the liquidated portfolio.  As a common strategy one chooses splitting of the given order into smaller pieces and to trade these pieces sequentially over time. The author can find optimal strategies explicitly and study their properties. The strategies depend on time and parameters of the model only and are solutions of a linear ordinary differential equation of the second order. The author proves that this case is a generalization of the original Merton's model studied in \cite{Merton1969}. He also noticed that the explicit and simple results for optimal strategies were possible to obtain just by using finite terminal time and because of the investor used the exponential  utility function. A more realistic setting, for instance where the investor receives multiply orders at random times or the liquidation time is not fixed in the beginning leads to an essential more complicated model. In comparison to our case the illiquid asset in \cite{Monin2014} does not pay any dividends and the investor can also split the illiquid asset and sell them piece by piece as well as the investor has no consumption during the lifetime of the portfolio.

Study of optimization problems with three assets including an illiquid asset leads to three dimensional nonlinear Hamilton--Jacobi--Bellman (HJB) equations. The corresponding nonlinear three dimensional partial differential equations (PDEs) include a lot of parameters describing the behavior of assets and are challenging for analytical and numerical methods. To simplify the investigated problem one tries to find an inner symmetry of such equation and reduce the number of independent variables at least to two or if possible to one. Usually, low dimensional problems better studied and are, therefore, easier to handle. We use in this paper the powerful method of Lie group analysis. This method is very well known for more than 100 years. It is very often used in the area of mathematical physics and over the past 20 years also in financial mathematics \cite{Bordag2015}. Nearly all known explicit solutions for ODEs or PDEs were found or can be found algorithmically by this method. This method is up to now the most appropriate method to find algorithmically substitutions to reduce a high dimensional PDE to a lower dimensional one or even to an ODE. The Lie algebraic structure of the corresponding HJB equation is one of the two main results of this paper. The Lie algebraic structure of the HJB equation reveals important structural properties of the considered equation. For instance, for the linear Black-Scholes equation the corresponding Lie algebraic structure gives rise to famous substitutions which reduce it to the heat equation or it allowed to obtain the fundamental solution and the explicit formulas for European Call- or Put. For a nonlinear equation like studied here we cannot get any fundamental solution but we can obtain reductions of the high dimensional HJB equation to simpler ones. We study the complete set of all possible reductions and describe the unique reduction to two dimensional PDE which satisfies the boundary conditions. We present also the explicit form of the corresponding investment-consumption strategies in invariant variables.This reduction also means that for all further investigation it is sufficient to use the two dimensional PDE instead of the three dimensional main HJB equation.

Our paper is organized as follows. In Section 2 we introduce the economic problem setting in detail. There is also provided a theorem that the HJB equation with a HARA utility function possesses a unique viscosity solution which was earlier proved in \cite{BordagYamshchikovZhelezov}.  It will be now adapted to the case exponential utility function. In Section 3 we provide the Lie group analyses of the optimization problems with a general liquidation time distribution and different utility functions. We prove that the cases HARA utility and exponential utility are completely different cases. The usual limiting procedure between the HARA and exponential utility functions gives us the wrong results for the corresponding Lie algebraic structures. In Section 4 for the optimization problem with the exponential utility function we chose an optimal system of subalgebras of the admitted Lie algebra and provide the complete set of all invariant reductions of the corresponding three dimensional PDE. Because this is an essential step, the complete prove with all details is given and the meaning of each reduction is explained. In each case we prove if the invariant substitutions are compatible with the boundary condition.
In Section 5 we discuss the connection between different results and see that the radical change of the investment –consumption strategies is connected with the chosen exponential utility function.

\section{Economical problem setting}
\label{sec:1}
 We study an optimization problem for a portfolio in the framework of continuous time. An investor has a portfolio with three assets: an illiquid, a liquid risky,  and a risk-free asset. The investor has an illiquid asset that has some paper value and can not be sold until some moment $T$ that is random with a prescribed liquidation time distribution. The investor tries to maximize her average consumption investing into a liquid risky asset that is partly correlated with the illiquid one. The investor is free to choose a utility function in correspondence with her risk tolerance. In our previous papers  \cite{BordagYamshchikovZhelezov}, \cite{Bordag2016}, \cite{BordagYamshchikov2017} we assumed that the investor  chooses a hyperbolic absolute risk aversion (HARA) utility function or a logarithmic (LOG) utility function as a special case of the HARA utility. Now we suppose that the investor has a  quite different risk tolerance as before and chooses an exponential utility function. We notice that a risk tolerance $R(c)$ of an investor is defined as $R(c)=-\frac{U'(c)}{U''(c)}$ for any utility function $U(c)$. For the HARA  utility function the risk tolerance $R(c)$ is a linear function of $c$ and for the exponential utility function it is a constant.

Because the form of utility functions define the form of HJB equations and the limiting procedures connecting the HARA utility with the logarithmical and exponential utility functions play an important role in this paper we describe these relations in more detail. In the previous papers  \cite{BordagYamshchikovZhelezov}, \cite{Bordag2016}, \cite{BordagYamshchikov2017} and \cite{BordagYamshchikovBook} we used the HARA  and LOG  utility functions and studied connection between both optimization problems. We used the HARA utility function in the form
\begin{equation}\label{harautility}
		U_1^{HARA}(c) = \frac{1-\gamma}{\gamma}\left( \left(\frac{c}{1-\gamma}\right)^\gamma - 1\right), ~~~ 0<\gamma <1.
	\end{equation}
 It is easily to see that as $\gamma \to 0$ then the  HARA utility function written as (\ref{harautility}) tends the to the LOG utility
\begin{equation} \label{HARAlimLOG}
{U_1^{HARA}(c)}_ { \gamma \to 0} {\longrightarrow} U^{LOG}(c)=\ln{c}.
\end{equation}
In common literature is often noticed that we obtain an exponential utility function as a limit case of  a HARA utility function by $\gamma \to \infty$. This assertion is correct just if the HARA utility function takes a special form, for instance, for the HARA utility in the form (\ref{harautility}) it is not the case. It is easy to prove that if we take the HARA utility in the form
\begin{equation}\label{harautility2}
		U_2^{HARA}(c) = \frac{1-\gamma}{\gamma}\left( \frac{a c}{1-\gamma} + 1\right)^\gamma , ~~~ 0<\gamma <1, ~~a>0,
	\end{equation}
then we obtain by the limiting procedure an exponential utility function
\begin{equation}\label{expnegative}
{U_2^{HARA}(c)}_{\gamma \to \infty} {\longrightarrow}  U^{EXPn}(c)=-e^{-a c}.
\end{equation}
 It is so called negative utility function (denoted  as EXPn). The most common form of the the exponential utility function is
\begin{equation}\label{exppositive}
U^{EXPp}(c)=\frac{1}{a}(1-e^{-a c}), ~a >0.
\end{equation}
We call it a positive exponential utility function (and denote it as EXPp). Both negative and positive exponential utility functions differ just by an additive and a multiplicative constant $\frac{1}{a}$. We will study later both optimization problems, with EXPn and EXPp utility functions.

The  forms  $U_1^{HARA}$ and $U_2^{HARA}$ of the HARA utility function are often used and from an economical point of view both of them have properties of a HARA type utility function. From the analytical point of view the HARA utility functions $U_1^{HARA}$ and $U_2^{HARA}$ are different.

In the case of $U_1^{HARA}$ we obtain by limiting transition $\gamma \to 0$ a LOG utility function (\ref{HARAlimLOG}). As we mentioned before the risk tolerance  in this case is equal to
\begin{equation}
R_1(c)=-\frac{U'(c)}{U''(c)}=\frac{c}{1-\gamma}  \label{risktolerance1},
\end{equation}
  and by $\gamma \to 0$ we obtain $R_{LOG}(c)=c$ as it to expect. But we get neither a finite limit by $\gamma \to \infty$ of the function $U_1^{HARA}(c)$ nor a relevant value for the risk aversion $R_1(c)$ in this case.

  In the second case for the utility function $U_2^{HARA}$ (\ref{harautility2}) we  obtain for the risk tolerance the expression
  \begin{equation}
R_2(c)=\frac{a c+ 1-\gamma}{a(1-\gamma)} =\frac{c}{1-\gamma} +\frac{1}{a}. \label{risktolerance2}
\end{equation}

Here  $U_2^{HARA}$ tends  by  the limiting transition $\gamma \to \infty$ to the EXPn utility function (\ref{expnegative}) and the risk tolerance takes a constant value $R_2(c)=a^{-1}$ as it is to expect in the case of an exponential utility function. But here in contradiction to the first case of $U_1^{HARA}$ we cannot obtain any meaningful expression by the limiting procedure  $\gamma \to 0$, it means we do not obtain a transition to the LOG utility function.

 In other words to study the connection between two optimization problems with a HARA utility function and with a logarithmic utility function we should use the HARA utility for instance in the form $U_1^{HARA}(c)$ to be able to provide the limiting procedure $\gamma \to 0$ in all formulas. For the study the connection between two optimization problems with a HARA utility and with an exponential utility we should use another form of the HARA utility, for instance, of type $U_2^{HARA}(c)$ to be able to make the limiting transition for $\gamma \to \infty$ in corresponding formulas.

Because of the relation (\ref{expnegative}) to correct comparison of the results for the optimization problem with the HARA utility function $U_2^{HARA}(c)$  with the results for an optimization problem with an exponential utility function we need to study first the optimization problem with the negative exponential utility function $U^{EXPn}(c)$.
From the other side the positive form of the utility function (\ref{exppositive}) is often used and it may be interesting to provide results also for this form of the utility function. The optimization problems with both utility functions describe economically equivalent situations. We decided to study both cases because of both types exponential utility functions are widely used. We will prove that the optimization problems with the EXPn and EXPp utility functions possesses also equivalent analytical and Lie algebraic structures. There exist one-to-one analytical substitution which provides the equivalence relation between two of these optimization problems which we provide explicitly later in Section \ref{positivenegative}.

We obtain  for the HARA utility function $R(c)=\frac{c}{1-\gamma}$, $ 0< \gamma<1$ where $\gamma$ is a parameter of the HARA utility function (given in the form introduced later in (\ref{harautility})), and $R(c)=c$ for the LOG utility function. It means also that the absolute risk tolerance $R(c)$ is increasing or decreasing with the consumption $c$ in these two cases.

For an exponential utility function (independently for the positive exponential (EXPp) utility function or for the negative exponential (EXPn) utility function) the risk tolerance $R(c)$ is a constant. In other words, all the time the absolute risk tolerance stays unaltered in the framework of these optimization problems.
We see that even though both the LOG  and EXPn utility functions can be regarded as a limit case of the HARA  utility function they describe quite different economical situations: in the first case, the risk tolerance changes with the consumption $c$ and in the case of the EXPn or EXPp utility functions the risk tolerance do not depend on the level of the consumption at all.

It means that now the investor has a constant risk tolerance. Maybe it explains that both optimization problems studied before and presented now have quite different analytic and  Lie algebraic structures as we show it later.

\subsection{Formulation of the optimization problem} \label{problem}

The investor's portfolio includes a risk-free bond $B_t$, a risky asset $S_t$ and a non-traded asset $H_t$ that generates stochastic income, i.e. dividends or costs of maintaining the asset. The liquidation time of the portfolio $T$ is a randomly-distributed continuous variable.
The risk free bank account $B_t$, with the interest rate $r$, follows
	\begin{equation} \label{bond_r}
		dB_t = rB_t\,dt, \, t \leq T,
	\end{equation}
where $r$ is constant. The lower case index $t$ denotes the spot value of the asset at the moment $t$.

 The stock price $S_t$ follows the geometrical Brownian motion
		\begin{equation} \label{asset_S}
		dS_t = S_t(\alpha\, dt + \sigma\, dW_t^{1}), \, t \leq T,
	\end{equation}
with the continuously compounded rate of return $\alpha > r$ and the standard deviation $\sigma$.
The illiquid asset $H_t$, that can not be traded up to the time $T$ and its paper value is correlated with the stock price and is governed by
\begin{equation}
		\frac{dH_t}{H_t} = (\mu - \delta)\,dt + \eta\left(\rho\, dW_t^{1} + \sqrt{1 - \rho^2}\,dW^{2}_t \right), \, t \le T,
	\label{eq:H1}
	\end{equation}
 where $\mu$ is the expected rate of return of the risky illiquid asset, $(W_t^{1}, W_t^{2})$ are two independent standard Brownian motions, $\delta$ is the rate of dividend paid by the illiquid asset, $\eta$ is the standard deviation of the rate of return, and $\rho \in(-1, 1)$ is the correlation coefficient between the stock index and the illiquid risky asset. The parameters $\mu$,  $\delta$, $\eta$, $\rho$ are all assumed to be constant.\\
The randomly distributed time $T$ is an exogenous time and it does not depend on the Brownian motions  $(W_t^{1}, W_t^{2})$. The probability density function of the liquidation time distribution is denoted by $\phi(t)$, whereas $\Phi(t)$ denotes the cumulative distribution function, and  $\overline{\Phi}(t)$, the survival function, also known as a reliability function, $\overline{\Phi}(t) = 1 - \Phi (t)$. We skip here the explicit notion of the possible parameters of the distribution in order to make the formulas shorter.
In dependence on the rate of illiquidity the liquidation time distribution can take different forms. Typically one use the simplest  one parameter exponential distribution with the reliability function $\overline{\Phi}(t)=e^{- \kappa t}$, where $\kappa$ is the parameter of the distribution or a more advanced  Weibull distribution with $\overline{\Phi} (t)=e^{-(t/\lambda)^k}$ with two parameters $\lambda$ and $k$. We will take these two distributions as examples in our investigation. We notice that the exponential  distribution is a special case of the Weibull distribution by $k=1$ and $\kappa =1/ \lambda$.

We assume that the investor consumes at rate $c(t)$ from the liquid wealth and the allocation-consumption plan $(\pi, c)$ consists of the allocation of the portfolio with the cash amount $\pi = \pi(t)$ invested in stocks, the consumption stream $c = c(t)$ and the rest of the capital kept in bonds. The consumption stream $c(t)$ is admissible if and only if it is positive and there exists a strategy that finances it. Further on we sometimes omit the dependence on $t$ in some of the equations for the sake of clarity of the formulas.  All the income is derived from the capital gains and the investor must be solvent. In other words, the liquid wealth process $L_t$ must cover the consumption stream.
The wealth process $L_t$ is the sum of cash holdings in bonds, stocks and {\em random} dividends from the non-traded asset minus the consumption stream, i.e. it  must satisfy the balance equation
\begin{eqnarray}
		 dL_t &=& \bigl(rL_t + \delta H_t + \pi_t(\alpha - r) - c_t\bigr)\, dt + \pi_t\sigma\, dW_t^1. \nonumber
	\end{eqnarray}
The investor wants to maximize the overall utility consumed up to the random time of liquidation $T$, given by
\begin{equation}
		\mathcal{U}(c) := E \left[\int^\infty_0{\overline{\Phi}(t)U(c)}\,dt\right], \label{eq:supinfty}
	\end{equation}
as it was shown in \cite{BordagYamshchikovZhelezov}.
It means we work with the problem (\ref{eq:supinfty}) that corresponds to the \emph{value function} $V(l,h,t)$, which is defined as
\begin{equation} \label{valueFun}
		V(l, h, t) = \max_{(\pi, c) } E \left[ \int_t^\infty \overline{\Phi}(t)U(c)\, dt \;|\; L(t) = l, H(t) = h \right],
	\end{equation}
where $l$ could be regarded as an initial capital and $h$ as a paper value of the illiquid asset. The value function $V(l, h, t)$ satisfies  the HJB equation
\begin{eqnarray}
		V_t (l, h, t) &+& \frac{1}{2}\eta^2h^2\,V_{hh} (l, h, t) + (r l + \delta h)\,V_l (l, h, t) + \nonumber \\
  		& & (\mu - \delta) h\,V_h (l, h, t) + \max_{\pi} G[\pi] + \max_{c \geq 0} H[c]\;\; =\;\; 0, \label{eq:HJB21}\\
		G[\pi] &=& \frac{1}{2}V_{ll}(l, h, t)\,\pi^2 \sigma^2 + V_{hl}(l, h, t)\,\eta\rho\pi\sigma h 
      		+ \pi(\alpha - r)\,V_l(l, h, t), \label{eq:Gmax21} \\
		H[c] &=& -c\,V_l (l, h, t) + \overline{\Phi}(t) U(c), \label{eq:Hmax21}
	\end{eqnarray}
with the boundary condition
\begin{equation} \label{cond:main}
		V(l, h, t) \to 0, ~~ as ~~ t \to \infty.
	\end{equation}
In \cite{Bordag2016} and \cite{BordagYamshchikovZhelezov} the authors have already demonstrated that the formulated problem has a unique solution under certain conditions. Namely, the following theorem was proved

\begin{theorem}\label{MT} \cite{BordagYamshchikovZhelezov}.
There exists a unique viscosity solution of the corresponding HJB equation (\ref{valueFun})-(\ref{cond:main}) if
\begin{enumerate}
\item $U(c)$ is strictly increasing, concave and twice differentiable in $c$,
\item $\lim_{t \to \infty} \overline{\Phi} (t) E[U(c(t))]=0$, $\overline{\Phi} (t) \sim e^{-\kappa t}$ or faster as $t \to \infty$,
\item $U(c) \leq M(1+c)^{\gamma}$ with $0<\gamma<1$ and $M>0$,
\item $\lim_{c \to 0} U'(c) = + \infty$,  $lim_{c \to + \infty} U'(c) = 0$.
\end{enumerate}
\end{theorem}

In this paper we restrict ourselves to the case of exponential utility functions that satisfy three first conditions of Theorem \ref{MT} by definition. We checked the proof of the theorem in \cite{BordagYamshchikovZhelezov} and see that the condition $\lim_{c \to 0} U'(c) \to \infty$  can be replaced by the condition  $\lim_{c \to 0} U'(c) >0$ and the existence and uniqueness of the viscosity solution of HJB equation is still  guaranteed.

Further we will use instead the fourth condition in Theorem \ref{MT}  the condition
\begin{equation}\label{expconditionUtility}
 U'(c)_{|_{c=0}} >0, ~~~~~ \lim_{c \to + \infty} U'(c) = 0
\end{equation}
which will be satisfied by negative or positive exponential utility functions. It means that there exist a unique viscosity solution to the HJB equation (\ref{valueFun})-(\ref{cond:main}) with an exponential utility function.

Now we can adjust and reformulate Lemma proved in \cite{BordagYamshchikovZhelezov} about the properties of the value function as follows

\begin{lemma}\label{PVF}
Under the conditions $(1) - (3)$ from Theorem \ref{MT} and the condition(\ref{expconditionUtility}) the value function $V(t, l, h)$ (\ref{valueFun}) has the following properties:
\begin{enumerate}
\item[\textup{(i)}] $V( l, h, t)$ is concave and non-decreasing in $l$ and in $h$,
\item[\textup{(ii)}] $V( l, h, t)$ is strictly increasing in $l$,
\item[\textup{(iii)}] $V( l, h, t)$ is strictly decreasing in $t$ starting from some point,
\item[\textup{(iv)}] $0 \leq V(l, h, t) \leq O(|l|^{\gamma} + |h|^{\gamma})$ uniformly in $t$.
\end{enumerate}
\end{lemma}

In the next sections,  we will first study three dimensional PDE which we obtain from the HJB equation after formal maximization, then we will try to simplify this three dimensional PDE as far as possible using its internal algebraic structure.  The properties of the value function listed in Lemma \ref {PVF}  we will use to define the reduction which keeps all properties of the original optimization problem. It follows that if one can find a solution to the reduced equation it will be also the unique viscosity solution of the optimization problem.

\section{Lie group analyses of the optimization  problem with a general liquidation time distribution and an exponential utility function}\label{negativeExp}

First we study the case of an optimization problem with  the EXPn utility function (\ref{expnegative}). As usual we provide a formal maximization of (\ref{eq:Gmax21}) and (\ref{eq:Hmax21}) for the chosen utility function in the HJB equation (\ref{eq:HJB21}) and get a three dimensional nonlinear PDE.

The HJB equation (\ref{eq:HJB21}) after the formal maximization procedure  will take the form
\begin{eqnarray}
		 && V_t(l, h, t) +  \frac{1}{2}\eta^2h^2V_{hh} ( l, h, t) + (r l + \delta h)V_l ( l, h, t) + (\mu - \delta) hV_h ( l, h, t) \nonumber \\
		 &-& \frac{(\alpha - r)^2 V_{l}^2(l,h, t)+2(\alpha - r)\eta \rho h V_{l}(l,h, t) V_{lh}(l,h, t) + \eta^2 \rho^2 \sigma^2 h^2 {V_{lh}}^2 (l,h, t)}{2 \sigma^2 V_{ll}(l,h, t)} \nonumber \\
 		&+& \frac{1}{a}V_{l}(l,h, t) \ln{V_{l}(l,h, t)}  -  \frac{1}{a}\left(1+\ln{ \overline{\Phi}(t)} \right)V_{l}(l,h, t) -\frac{\ln a}{a} ~ V_{l}(l,h, t)  =0 , \label{maingeneralExpNegative}\\
&& ~~~~~~~~~~~~~V \rightarrow 0, ~  t \longrightarrow \infty.\nonumber
	\end{eqnarray}
Here the investment $\pi( l , h, t)$ and consumption $c( l, h, t)$ strategies look as follows in terms of the value function V(l,h, t)
\begin{eqnarray}
		 \pi ( l ,h, t) &=& - \frac{\eta \rho \sigma h V_{lh}( l, h, t) + (\alpha - r) V_l ( l, h, t)}{\sigma^2 V_{ll}(l,h, t)}, \label{pi:maingeneralExpNegative}\\
		 c( l, h, t) &=&\frac{1}{a} \ln{\left( \frac{ \overline{\Phi}(t)}{a V_l ( l,h, t)} \right)}. \label{c:maingeneralExpNegative}
\end{eqnarray}
Equation (\ref{maingeneralExpNegative}) is a nonlinear three dimensional PDE with the three independent variables $l,h,t$ such equations are demanding by study with analytical or numerical methods.
The Lie group analysis of a nonlinear PDE is a proper tool to obtain the Lie algebra admitted by this equation. Using the  generators of this  symmetry algebra one can reduce the dimension of the equation (\ref{maingeneralExpNegative}) and make a problem better tractable.

Roughly speaking  to obtain internal Lie algebraic structure of a differential equation on any function $V(l,h,t)$ we present this equation as an algebraic equation, for instance like $\Delta(l,h,t,V,V_l,V_h,V_t,V_{ll},V_{ll},V_{lh},V_{hh})=0$ in some special space called jet bundle.  This space is denoted by $j^{(n)}$, where $n$ is the order of the highest derivative in the differential equation.  All derivatives will be now considered as new dependent variables.  Thereafter we study the properties of the solution manifold of this equation, which is now a surface in the jet bundle $j^{(n)}$. We take a generator $ \mathbf U$ of a point transformation in the corresponding jet bundle and act on the solution manifold to define invariant subspaces.  One obtains a large system of partial differential equations on the coefficients of the generator $ \mathbf U$.  Usually this system does not has any nontrivial solution at all, and correspondingly the studied differential equation does not admit any symmetry. In seldom cases, one gets nontrivial generators of the point transformations admitted by the equation.  The symmetry properties will then used to simplify the studied equation and one obtains a so called reduced equation. In detail one can find the description of this method in \cite{Olver}, \cite{Ibragimov1994} or in \cite{Bordag2015} where a short and comprehensive introduction in this method is given as well as applications to other PDEs arising in financial mathematics.

Here we formulate the main theorem of Lie group analysis for the optimization problem with the EXPn utility function.

 \begin{theorem}\label{THEXPn}
The HJB equation (\ref{maingeneralExpNegative}) with the EXPn utility function (\ref{expnegative}) and with a general liquidation time distribution ${\Phi} (T)$ admits the four dimensional Lie algebra $L^{EXPn}_4$ spanned by generators $\mathbf U_1, \mathbf U_2, \mathbf U_3, \mathbf U_4$, i.e. $L^{EXPn}_4=<\mathbf U_1, \mathbf U_2, \mathbf U_3, \mathbf U_4>$, where
\begin{eqnarray} \label{symalgenaralExpNegativ}
&& \mathbf U_1 =  \frac{1}{ar}\frac{\partial}{\partial l} -  V \frac{\partial}{\partial V},~~~~~~~~~~~~~~~~~~~~~~~~~~~~~~~~~\mathbf U_2 = \frac{\partial}{\partial V}, \\
&& \mathbf U_3 =  -\frac{1}{a r}\left(e^{r t}\int{e^{-r t}d{\ln{\overline{\Phi}(t)}}}  \right)\frac{\partial}{\partial l} +  \frac{1}{r}\frac{\partial}{\partial t},  ~~~~~~\mathbf U_4 = e^{rt}\frac{\partial}{\partial l}\nonumber.
\end{eqnarray}
with following nontrivial commutation relations
\begin{equation}\label{nonzeroComEXPn}
[\mathbf U_1,\mathbf U_2]= \mathbf U_2, ~~~~~[\mathbf U_3,\mathbf U_4]= \mathbf U_4.
 \end{equation}
 Except finite dimensional Lie algebra $L^{EXPn}_4$ (\ref{symalgenaralExpNegativ}) the equation (\ref{maingeneralExpNegative}) admits also an infinite dimensional algebra $L_{\infty}=<\psi(h,t) \frac{\partial}{\partial V}>$ where the function $\psi(h,t)$ is any solution of the linear parabolic PDE
\begin{equation}\label{inftyAlExpNegative}
 		\psi_t( h,t) +  \frac{1}{2}\eta^2h^2 \psi_{hh} ( h,t)  + (\mu - \delta) h \psi_{h} ( h,t)=0.
	\end{equation}
 \end{theorem}

\begin{proof}
As in \cite{Olver}, \cite{Ibragimov1994} or \cite{Bordag2015} we introduce the second jet bundle $j^{(2)}$ and present the equation (\ref{maingeneralExpNegative}) in the form
$\Delta(l,h,t,V,V_l,V_h,V_t,V_{ll},V_{ll},V_{lh},V_{hh})=0$ as a function of these variables in the jet bundle $j^{(2)}$. We look for generators of the admitted Lie algebra in the form
	\begin{equation}\label{operatorU}
		\mathbf U=\xi_1(l,h,t,V)\frac{\partial}{\partial l}+\xi_2(l,h,t,V)\frac{\partial}{\partial h} +\xi_3(l,h, t, V)\frac{\partial}{\partial t}+\eta_1(l,h, t, V)\frac{\partial}{\partial V},
	\end{equation}
where the functions $\xi_1,\xi_2,\xi_3,\eta_1$ can be found using the over determined system of determining equations
	\begin{equation}\label{operatorU2}
		\mathbf U^{(2)}\Delta(l,h, t, V,V_l,V_h,V_t,V_{ll},V_{ll},V_{lh},V_{hh})|_{\Delta=0}=0,
	\end{equation}
where $\mathbf U^{(2)}$ is the second prolongation of  $\mathbf U$ in $j^{(2)}$. We look at the action of $\mathbf U^{(2)}$ on   $\Delta(l,h, t, V,V_l,V_h,V_t,V_{ll},V_{ll},V_{lh},V_{hh})$ located  on its solution subvariety $\Delta=0$ and obtain an overdetermined system of PDEs on the functions $\xi_1$, $\xi_2$, $\xi_3$ and $\eta_1$ from (\ref{operatorU}).
This system has 130 PDEs on the functions $\xi_1,\xi_2,\xi_3,\eta_1$. The most of them are trivial and lead to following conditions on the functions
	\begin{eqnarray}
		(\xi_1)_h &=&0,~(\xi_1)_V = 0,~(\xi_1)_{ll} = 0, ~(\xi_1)_{l} = \xi_{11}(t),\nonumber \\
		(\xi_2)_l &=& 0,~(\xi_2)_V = 0, \nonumber \\
		(\xi_3)_l &=& 0,~(\xi_3)_h = 0,~(\xi_3)_V = 0,~ \nonumber \\
		(\eta_1)_l &=& 0,~(\eta_1)_{VV} =0, ~(\eta_1)_V = \eta_{11}(h,t). \nonumber
	\end{eqnarray}
Consequently, the unknown functions in (\ref{operatorU}) have the following structure
	\begin{eqnarray}\label{condExpNegative}
		\xi_1(l, h, t, V) &=& \xi_{11}(t) l +\xi_{12}(t), ~~\xi_2 (l, h, t, V)=\xi_2(h,t),~~\xi_3 (l, h, t, V)=\xi_3(t), \nonumber \\
		\eta_1 (l, h, t, V)&=&\eta_{11}(h,t) V+ \eta_{12}(h,t).
	\end{eqnarray}
Here $\xi_{11}(t), \xi_{12}(t), \xi_2(t,h),\xi_{3}(t) , \eta_{11}(t,h), $ and $\eta_{12}(t,h)$ are some functions which will be defined later.
To find these unknown functions  we should have a closer look at the nontrivial equations of the obtained system, that are left. After all simplifications, we get the system of seven PDEs
\begin{eqnarray} \label{lieSystemExpNegativ}
{\eta_1}_t+\frac{\eta^2 h^2}{2}{\eta_1}_{hh}+(\mu-\delta)h {\eta_1}_h& = & 0,  \\
 {\xi_3}_t -{\xi_1}_l&=& 0, \nonumber \\ 
 (\mu - \delta) ({\xi_2} - h{\xi_2}_h +h{\xi_3}_t) -\frac{1}{2} \eta^2 h^2 {\xi_2}_{hh} +\eta^2 h^2 {\eta_{11}}_{h} &=& 0, \nonumber \\ 
 h {\xi_3}_t +2(\xi_2-h {\xi_2}_h)&=& 0, \nonumber \\
 \frac{1}{a}\eta_{11} +r \xi_1 - \frac{1}{a} {\xi_1}_t+\delta \xi_2-\frac{1}{a}\frac{\overline{\Phi}_t}{\overline{\Phi}}{\xi_3}&=&0 \nonumber\\ 
  (\alpha- r){\xi_3}_t +2\eta \rho h {\eta_{11}}_{h}&=& 0, \nonumber\\
  (\alpha - r)({\xi_2} - h{\xi_2}_h + h{\xi_3}_t) +\eta \rho \sigma^2 h^2 {\eta_{11}}_{h} &=& 0. \nonumber  
\end{eqnarray}
  We introduce the differential operator $\textbf{L} = \frac{\partial}{\partial t} + \frac{1}{2} \eta^2 h^2 \frac{\partial^2}{\partial h^2} +(\mu-\delta)h \frac{\partial}{\partial h}$ using this operator we can rewrite the first equation in the above system as conditions on the functions $\eta_{11}(h,t)$ and $\eta_{12}(h,t)$  which appears in the last equation of (\ref{condExpNegative}) correspondingly as $\textbf{L}\eta_{11}(h,t)=0$ and $\textbf{L}\eta_{12}(h,t)=0.$ Other equations in the above system do not contained the function $\eta_{12}(h,t)$ at all. If we now denote $\eta_{12}(h,t)=\psi(h,t)$ then we see that we proved the last statement of the theorem, see (\ref{inftyAlExpNegative}).

Solving the  system (\ref{lieSystemExpNegativ})  for an arbitrary function $\overline{\Phi}(t)$ we obtain
\begin{eqnarray} \label{lieSystemEXPn}
\xi_1 = c_{11} e^{r t} + \eta_{11} \frac{1}{ar} -\xi_3 \frac{1}{a} e^{rt}\int{ e^{-rt}  \frac{{\overline{\Phi}}_t}{\overline{\Phi}}dt}, ~~~~~~~~~~~
\xi_2 = 0,~~~~~~~~~~~~~~~~~~~~~\nonumber\\
\xi_3 - const.,~~~~~~\eta_1 = \eta_{11} V + \eta_{122}+\psi(h,t), ~~~~~~~~c_{11},\eta_{11},\eta_{122}-const.~~~~
\end{eqnarray}
The equations (\ref{lieSystemEXPn}) contain four arbitrary constants $\xi_3, c_{11},\eta_{11},\eta_{122}$ and a function $\eta_{12}(h,t)=\psi(h,t)$ which is an arbitrary solution of $\textbf{L}\psi(h,t) = 0$.
Formulas (\ref{lieSystemEXPn}) define four generators of the finite dimensional Lie algebra $L^{EXPn}_4$ (\ref{symalgenaralExpNegativ}) and the infinitely dimensional algebra $L_\infty$ (\ref{inftyAlExpNegative}) as it was described in Theorem~\ref{THEXPn}.
\end{proof} \qed

\begin{remark}
The found four dimensional Lie algebra describes the symmetry property of the equation (\ref{maingeneralExpNegative}) for any function $\overline{\Phi} (t)$.  In \cite{BordagYamshchikovZhelezov}, \cite{Bordag2016}   we have proved the theorem for existence and uniqueness of the solution of HJB equation for a liquidation time distribution for which $\overline{\Phi} (t) \sim e^{-\kappa t}$ or faster as $t \to \infty$, therefore we will regard this type of the distribution studying the analytical properties of the equation further on.
\end{remark}

First we explain the meaning of some generators of the Lie algebra listed in Theorem \ref{THEXPn}.
We start with the second generator $\mathbf U_2 = \frac{\partial}{\partial V}$. It means that the original value function $V(l,h,t)$ which is a solution of the equation (\ref{maingeneralExpNegative}), can be shifted on any constant and still be a solution of the same equation. Neither allocation $\pi$ or consumption function $c$ will change their values because they also depend only on the derivatives of the value function. In some sense it is a trivial symmetry, since the equation (\ref{maingeneralExpNegative}) contains just the derivatives of $V(l,h,t)$ so we certainly can add a constant to this function and it still will be a solution of the equation. Following this symmetry does not give a rise to any reductions of the studied three dimensional PDE and this symmetry do not satisfy the boundary condition $V(l,h,t) \rightarrow 0, ~  t \longrightarrow \infty$ because of that it is not interesting by solving of the possed optimization problem.

The fourth generator $\mathbf U_4 = e^{rt}\frac{\partial}{\partial l}$ means that the value of the independent variable $l$ can be shifted on the arbitrary value $d e^{rt}$, i.e. the shift $l \to l+d e^{rt}$, $d-const.$ leaves the solution unaltered. From economical point of view it means that the absolut value of the initial capital is not important for this problem. We can arbitrary shift the initial liquidity $l$ on a bank account $d, d>0$ or credit $d, d<0$ yet $l+d e^{rt}$ should be positiv in the initial time moment. The value function $V(l,h,t)$ as a solution of the equation (\ref{maingeneralExpNegative}) and the allocation-consumption strategy $(\pi,c)$ will be unaltered. This symmetry is trivial and it does not provide any reductions of the original three dimensional PDE.

We also get the infinitely dimensional algebra $L_{\infty}=<\psi(h,t) \frac{\partial}{\partial V}>$ where the function $\psi(h,t)$ is any solution of the linear PDE
$\psi_t +  \frac{1}{2}\eta^2h^2 \psi_{hh}   + (\mu - \delta) h \psi_{h} =0 $, see Theorem \ref{THEXPn}. It has a special meaning -
we can add any solution $\psi(h,t)$ of this equation to the value function $V(l,h,t)$ without any changes of the allocation-consumption strategy $(\pi,c)$. From economical point of view it means that the additional use of some financial instrument which is the solution of  $ \psi_t +  \frac{1}{2}\eta^2h^2 \psi_{hh}   + (\mu - \delta) h \psi_{h} =0
$ do not change the investment- allocation strategies in this optimization problem. The boundary condition $V(l,h,t) \rightarrow 0, ~  t \longrightarrow \infty$ leads to the following boundary condition on the solution of this equation $\psi(h,t) \rightarrow 0, ~  t \longrightarrow \infty$.
It means it is a financial instrument which value is defined just by the paper value of the illiquid asset and time only, can not change the allocation-consumption strategy $(\pi,c)$.
We notice also that after the substitution $h=e^x,$ $t=-2\tau/\eta^2,$ $\psi(t,h)=h^{(\mu - \delta-\frac{1}{2}\eta^2)/\eta^2}e^{-(\mu - \delta-\frac{1}{2}\eta^2)^2 \tau/\eta^4} v(\tau,x)$ we obtain on the function $v(\tau,x)$ the parabolic  equation of the type $v_t=v_{xx}$, which is well studied. The solution methods as well as the fundamental solution of this equation are well known.

\subsection{Relation between two optimization problems: one  with the HARA and one with the EXPn utility function}\label{HARAtoEXPn}

In our previous papers we studied the optimization problem with an illiquid asset in the case if the investor used   the HARA utility function (\ref{harautility}) or the logarithmic utility function (\ref{HARAlimLOG}). It is well known  that both problems are connected by limiting procedure if $\gamma \to 0$. In the previous paper \cite{BordagYamshchikov2017} we proved that also analytic and Lie algebraic structures of both optimization problems also connected with the same limiting procedure.

We noticed before that the EXPn utility function is connected to the HARA utility function  $U_2^{HARA}(c)$  with the limiting procedure by $\gamma \to \infty$, see (\ref{expnegative}).

It means also that we cannot use directly the results  of the Lie group analysis  obtained in previous works \cite{BordagYamshchikov2017} and \cite{BordagYamshchikovBook} to compare the admitted Lie algebras for the optimization problem with the HARA utility function in the form $U_1^{HARA}(c)$ with the results in this work for an optimization problem with an exponential utility function. Because of that we should recalculate the results of the Lie group analysis for the new form of the HARA utility function.
We remember that we first provide the formal maximization in the HJB equation (\ref{eq:HJB21}) and correspondingly to the chosen utility function we obtain a three dimensional PDE. In our previous works \cite{BordagYamshchikov2017} and \cite{BordagYamshchikovBook} we used the utility function $U_1^{HARA}(c)$ in the form (\ref{harautility}) and got following PDE
\begin{eqnarray}\label{maingeneralHARA}
		 && V_t(t, l, h) +  \frac{1}{2}\eta^2h^2V_{hh} (t, l, h) + (rl + \delta h)V_l (t, l, h) + (\mu - \delta) hV_h (t, l, h)  \nonumber \\
		 &-& \frac{(\alpha - r)^2 V_{l}^2(t,l,h)+2(\alpha - r)\eta \rho h V_{l}(t,l,h) V_{lh}(t,l,h) + \eta^2 \rho^2 \sigma^2 h^2 {V_{lh}}^2 (t,l,h)}{2 \sigma^2 V_{ll}(t,l,h)} \nonumber \\
 		&+& \frac{(1-\gamma)^2}{\gamma}\overline{\Phi}(t)^{\frac{1}{1-\gamma}}V_{l} (t,l,h)^{-\frac{\gamma}{1-\gamma}}-\frac{1-\gamma}{\gamma}\overline{\Phi}(t)=0, ~~~~~ V_{t \to \infty}{\longrightarrow}  0.
	\end{eqnarray}
Now if we insert in  the HJB equation (\ref{eq:HJB21}) the HARA utility function $U_2^{HARA}$ (\ref{harautility2}) then we obtain the PDE in the form
\begin{eqnarray}\label{maingeneralHARA2}
		 && V_t(t, l, h) +  \frac{1}{2}\eta^2h^2V_{hh} (t, l, h) + (rl + \delta h)V_l (t, l, h) + (\mu - \delta) hV_h (t, l, h)  \nonumber \\
		 &-& \frac{(\alpha - r)^2 V_{l}^2(t,l,h)+2(\alpha - r)\eta \rho h V_{l}(t,l,h) V_{lh}(t,l,h) + \eta^2 \rho^2 \sigma^2 h^2 {V_{lh}}^2 (t,l,h)}{2 \sigma^2 V_{ll}(t,l,h)} \nonumber \\
 		&+& \frac{(1-\gamma)^2}{\gamma}\overline{\Phi}(t)^{\frac{1}{1-\gamma}}V_{l} (t,l,h)^{-\frac{\gamma}{1-\gamma}}-\frac{1-\gamma}{a} V_l (t, l, h)=0, ~~~~~ V_{t \to \infty}{\longrightarrow}  0.
	\end{eqnarray}
The equations (\ref{maingeneralHARA}) and (\ref{maingeneralHARA2}) differ analytically  in the last terms, from an economical point of view they describe  equivalent optimization problems. The Lie group analysis of the first equation was provided in \cite{BordagYamshchikov2017}. Now we use the same method and find the admitted Lie group for the equation (\ref{maingeneralHARA2}). We formulate the results of in the following theorem
\begin{theorem}\label{MT2}
	The equation (\ref{maingeneralHARA2}) admits the three dimensional Lie algebra $L^{HARA_2}_3$ spanned by generators $L^{HARA_2}_3=<\mathbf U_1, \mathbf U_2, \mathbf U_3>$, where
	\begin{eqnarray}
\mathbf U_1 = \frac{\partial}{\partial V}, ~~\mathbf U_2 = e^{rt} \frac{\partial}{\partial l},~~
		 \mathbf U_3 =\left( l +\frac{1-\gamma}{ar}\right)\frac{\partial}{\partial l} + h \frac{\partial}{\partial h} +\gamma V  \frac{\partial}{\partial V}, \label{symalharageneral2}
	\end{eqnarray}
for any liquidation time distribution. Moreover, if and only if the liquidation time distribution has the exponential form, i.e. $\overline{\Phi} (t) =d e^{-\kappa t}$, where $d,\kappa$ are constants  the studied equation 	admits a four dimensional Lie algebra $L^{HARA_2}_4$ with an additional generator
\begin{equation} \label{exp4sym}
	\mathbf U_4= \frac{\partial}{\partial t} -\kappa V \frac{\partial}{\partial V},
	\end{equation}
i.e. $L^{HARA_2}_4=<\mathbf U_1, \mathbf U_2, \mathbf U_3, \mathbf U_4>$.\\
Except finite dimensional Lie algebras (\ref{symalharageneral}) and (\ref{exp4sym}) correspondingly equation (\ref{maingeneralHARA2}) admits also an infinite dimensional algebra $L_{\infty}=<\psi(h,t) \frac{\partial}{\partial V}>$ where the function $\psi(h,t)$ is any solution of the linear PDE
\begin{equation}
 		\psi_t( h,t) +  \frac{1}{2}\eta^2h^2 \psi_{hh} ( h,t)  + (\mu - \delta) h \psi_{h} ( h,t)=0.
	\end{equation}
The Lie algebra $L^{HARA_2}_3$ has the following non-zero commutator relations
\begin{equation}
\left[\mathbf U_1, \mathbf U_3\right] =  \gamma \mathbf U_1,~~~~ \left[\mathbf U_2, \mathbf U_3\right] = \mathbf U_2
\end{equation}
The Lie algebra $L^{HARA_2}_4$ has the following non-zero commutator relations
\begin{eqnarray}
\left[\mathbf U_1, \mathbf U_3\right] =  \gamma \mathbf U_1, ~ \left[\mathbf U_1, \mathbf U_4\right] = -\kappa \mathbf U_1, \nonumber\\
 \left[\mathbf U_2, \mathbf U_3\right] = \mathbf U_2, ~~~~ \left[\mathbf U_2, \mathbf U_4\right] = -r \mathbf U_2.
\end{eqnarray}
\end{theorem}
We will not provide the proof of the Theorem \ref{MT2} because it is quite similar to the proof of the previous Theorem \ref{THEXPn} for the equation (\ref{maingeneralExpNegative}). In the paper \cite{BordagYamshchikov2017} we obtained for the equation (\ref{maingeneralHARA}) with a general liquidation time distribution another three dimensional Lie algebra $L_3^{HARA_1}$  with generators of the following form
\begin{eqnarray}
&&{\cal U}_1 = \frac{\partial}{\partial V}, ~~~~~~~~~~~~~~~~~{\cal U}_2 = e^{rt} \frac{\partial}{\partial l}, \nonumber\\
&&		 {\cal U}_3 = l \frac{\partial}{\partial l} + h \frac{\partial}{\partial h} + \left(\gamma V - (1-\gamma) \int \overline{\Phi}(t)  dt\right) \frac{\partial}{\partial V}, \label{symalharageneral}
	\end{eqnarray}
as well as an infinite dimensional algebra $L_{\infty}=<\psi(h,t) \frac{\partial}{\partial V}>$ where the function $\psi(h,t)$ is any solution of the linear PDE $
 		\psi_t( h,t) +  \frac{1}{2}\eta^2h^2 \psi_{hh} ( h,t)  + (\mu - \delta) h \psi_{h} ( h,t)=0.$
 It is easy to see that both algebras $L^{HARA_1}_3$ and $L^{HARA_2}_3$ have the same commutation relations and are isomorph. We prove that the admitted Lie algebras are also similar. Indeed if we take the substitutions
\begin{eqnarray}\label{transformationHara}
\tilde l =l-\frac{1-\gamma}{a r},~~ \tilde h=h,~~\tilde t=t,~~\tilde V=V +\frac{1-\gamma}{\gamma}\int{\overline{\Phi}(t)}dt
\end{eqnarray}
then the equation (\ref{maingeneralHARA}) on the function $\tilde V(\tilde l,\tilde h,\tilde t)$ will be replaced by the equation (\ref{maingeneralHARA2}) on the value function $V(l,h,t)$. Correspondingly the generators of the algebra (\ref{symalharageneral}) will take the form of the generators of  (\ref{symalharageneral2}). It means  that the Lie algebras $L^{HARA_1}_3$  and $L^{HARA_2}_3$  are isomorph and similar and the optimization problems with the utility functions $U_1^{HARA}$ and with $U_2^{HARA}$ are equivalent not only from an economical and an analytical  but also from the Lie algebraic point of view. Now we have a correct form of generators of the Lie algebra to study a limiting procedure by $\gamma \to \infty$. Indeed using the properties of the generators of the Lie algebra we obtain from (\ref{symalharageneral2})
\begin{eqnarray}
\mathbf U_1^{\infty} = \frac{\partial}{\partial V}, ~~\mathbf U_2^{\infty} = e^{rt} \frac{\partial}{\partial l},~~
		 \mathbf U_3^{\infty} =\frac{1}{ar}\frac{\partial}{\partial l}  - V  \frac{\partial}{\partial V}, \label{symalharageneral2infty}
	\end{eqnarray}
We apply now the limiting transition to the main equation (\ref{maingeneralHARA2}) and see that we do not get the corresponding PDE in the form (\ref{maingeneralExpNegative}), despite the fact that the $U_2^{HARA}$ and EXPn utility are connected with this limiting procedure. We get different analytical structures on this step.

 Now we compare this Lie algebraic structure  with the described in  Theorem \ref{THEXPn}. First we see that both three dimensional PDEs have the same infinite dimensional algebra $L_{\infty}=<\psi(h,t) \frac{\partial}{\partial V}>$. Then we compare the finite dimensional algebra (\ref{symalharageneral2infty}) with (\ref{symalgenaralExpNegativ}) and see that the finite dimensional algebras in these cases are essentially  different. In the case of the $U_2^{HARA}(c)$ utility function  and a general liquidation time distribution we have after limiting transition $\gamma \to \infty$ the three dimensional algebra (\ref{symalharageneral2infty}) and in the case of the EXPn utility function we got the four dimensional algebra (\ref{symalgenaralExpNegativ}). These algebras do not connect with the limiting procedure by $\gamma \to \infty$ as well as the both three dimensional PDEs   (\ref{maingeneralHARA2}) and  (\ref{maingeneralExpNegative}) are not connected with this limiting procedure. Other sides it is easy to see that  all three generators (\ref{symalharageneral2infty}) coincide with the three of the four generators of (\ref{symalgenaralExpNegativ}). Lie algebra (\ref{symalgenaralExpNegativ}) is in some way extension of the Lie algebra (\ref{symalharageneral2infty}). It means that using the exponential utility functions like EXPn or EXPp makes the corresponding optimization problem smoother from the Lie algebraic point of view.

We see that by limiting procedure $\gamma \to \infty$ neither the analytic nor the Lie algebraic structure of the optimization problem will be preserved. If in the previous cases of the HARA  and LOG  utility functions it was sufficient to study the case of the HARA utility and then just take a limit by $\gamma \to 0 $ to obtain the corresponding results for the optimization problem with LOG  utility function now we should study the optimization problem with the exponential utility function in own rights  step-by-step independently from the case of the HARA  utility function.

In the next section we study the optimization problem with the EXPp utility function. Because the utility functions are defined up to additive and multiplicative constants the results should be in some sense equivalent. We will prove this equivalence explicitly for the optimization problems with the EXPn and EXPp utility functions in the next section.

\subsection{Relation between two optimization problems correspondingly one with the EXPn and one with the EXPp utility function}\label{positivenegative}

  The EXPp utility function (\ref{exppositive}) is very close to the  EXPn function (\ref{expnegative}), yet this particular case rather popular therefore we analyze it separately.

The whole approach is very similar to the method described at the beginning of this section therefore we omit some details here. In the case of the EXPp utility function (\ref{exppositive}) the HJB equation after the formal maximization procedure will take the following form
\begin{eqnarray}\label{maingeneralExpPositive}
		 && V_t  +  \frac{1}{2}\eta^2h^2V_{hh}  + (rl + \delta h)V_l + (\mu - \delta) h V_h   \\
		 &-& \frac{(\alpha - r)^2 V_{l}^2 + 2(\alpha - r)\eta \rho h V_{l} V_{lh} + \eta^2 \rho^2 \sigma^2 h^2 {V_{lh}}^2 }{2 \sigma^2 V_{ll}} \nonumber \\
 		&+& \frac{1}{a}V_{l} \ln{V_{l}}  -  \frac{1}{a}\left( 1+\ln{ \overline{\Phi}(t)} \right) V_{l}   + \frac{1}{a}\overline{\Phi}(t)=0 ,~~~~~V \rightarrow 0, ~  t \rightarrow \infty.\nonumber
	\end{eqnarray}
The main PDE  (\ref{maingeneralExpNegative}) for the EXPn utility function differs from this one for the EXPp utility function by one term only. If in (\ref{maingeneralExpNegative}) the last term was $-\left(\frac{1}{a}\ln{a}\right) V_l$ now it is $\frac{1}{a}\overline{\Phi}(t)$, i.e. now equation (\ref{maingeneralExpPositive}) has a free term $\frac{1}{a}\overline{\Phi}(t)$ without the dependent variable $V(l,h,t)$ as it was before.

Analogously to the previous chapter we can formulate and prove the main theorem of Lie group analysis for the HJB optimization problem with the positive exponential utility function.

\begin{theorem}\label{TEXPp}
	The equation (\ref{maingeneralExpPositive}) admits the four dimensional Lie algebra $L^{EXPp}_4$ spanned by generators $\mathbf U_1, \mathbf U_2, \mathbf U_3, \mathbf U_4$, i.e.  $L^{EXPp}_4=<\mathbf U_1, \mathbf U_2, \mathbf U_3, \mathbf U_4>$, where
	 \begin{eqnarray} \label{symalgenaralExpPositiv}
&& \mathbf U_1 =  \frac{1}{a r} \frac{\partial}{\partial l} - \left(   V+ \frac{1}{a} \int{\overline{\Phi}(t)dt} \right) \frac{\partial}{\partial V},~~~~~~~~~~~~~~~~~~~~~~~\mathbf U_2 = \frac{\partial}{\partial V},\\
&& \mathbf U_3 =  -\frac{1}{ar}e^{r t}\left(\int{e^{-r t}d{\ln{\overline{\Phi}(t)}}}  \right)\frac{\partial}{\partial l} + \frac{1}{r} \frac{\partial}{\partial t} - \frac{1}{a r}\overline{\Phi}(t) \frac{\partial}{\partial V},~~~~~~\mathbf U_4 = e^{rt}\frac{\partial}{\partial l}.\nonumber
\end{eqnarray}
for any liquidation time distribution $\overline{\Phi}(t)$.
Except finite dimensional Lie algebra $L^{EXPp}_4$ the equation (\ref{maingeneralExpPositive}) admits also an infinite dimensional algebra $L_{\infty}=<\psi(h,t) \frac{\partial}{\partial V}>$ where the function $\psi(h,t)$ is any solution of the linear PDE
\begin{equation}
 		\psi_t( h,t) +  \frac{1}{2}\eta^2h^2 \psi_{hh} ( h,t)  + (\mu - \delta) h \psi_{h} ( h,t)=0.
	\end{equation}
The Lie algebra $L^{EXPp}_4$ has following two non-zero commutator relations
\begin{equation} \label{nonzeroConEXPp}
\left[\mathbf U_1, \mathbf U_2\right] =  \mathbf U_2,~~~ \left[\mathbf U_3, \mathbf U_4\right] = \mathbf U_4.
\end{equation}
\end{theorem}

If we compare nontrivial commutation relations (\ref{nonzeroComEXPn}) for $L_4^{EXPn}$ and (\ref{nonzeroConEXPp}) for $L_4^{EXPp}$ we see that they coincide. It means that both algebras $L^{EXPp}_4$ and $L^{EXPn}_4$ are isomorph because they have the same set of the structure constants.  It means also that both Lie symmetry algebras $L^{EXPn}_4$ and $L_4^{EXPp}$ correspond to the type $A_2 \oplus A_2$ after the classification provided in \cite{paterawintern}.

If we can prove that the both Lie algebras are also similar, then  both optimization problems with the EXPp and with the EXPn utility functions are equivalent and we provide at the same time the equivalence substitution.
The similarity of two algebras means that they are not just isomorphic but also  that there exists an analytical substitution that  provided analytical identity between corresponding generators. We can then use this substitution to transform the equation (\ref{maingeneralExpPositive}) to (\ref{maingeneralExpNegative}) or both equations to one the same equation.

It is easy to see that if we make following transformations of the variables $l,h,t,V$ in equation (\ref{maingeneralExpPositive})
\begin{eqnarray}\label{transform}
l=\tilde l -\frac{\ln a}{a r},~~~h=\tilde h,~~~t=\tilde t, \\
V(l,h,t)=\tilde V(\tilde l,\tilde h,\tilde t)-\frac{1}{a}\int{\overline{\Phi}(t)}dt. \nonumber
\end{eqnarray}
then the final equation in variables $\tilde l,\tilde h,\tilde t, \tilde V$ coincide with (\ref{maingeneralExpNegative}).
Because the substitution (\ref{transform}) is an invertible analytical one-to-one substitution we have to do with two identical optimization problems. The analytical and the Lie algebraic structures of the optimization problems with the EXPn and EXPp utility functions are equivalent and it is enough to study one of these problems in detail.

\begin{remark} The EXPp and the EXPn utility functions are connected by an affine transformation. Because of that, it seems that we do not need the Lie group analysis to define transformations like in  (\ref{transform}). If we look at the formula (\ref{valueFun}) we can guess that at liest the transformation of the value function should look like in (\ref{transform}). But the maximization procedure used in formula (\ref{valueFun})  is non-trivial and not all transformations applied to the utility function will be preserved after this procedure. For instance it is rather difficult just by looking at the formula (\ref{valueFun})   to see that the liquidity variable should be also shifted.
 Other sides we have seen that the utility functions $U_1^{HARA}$ and  $U_2^{HARA}$ are connected with an algebraic relation but the corresponding PDEs are connected with the linear transformation (\ref{transformationHara}) which is very similar to (\ref{transform}). In both cases, the Lie group analysis gives us the correct answer to how the corresponding optimization problems are connected.
\end{remark}

\section{Optimal system of subalgebras of $L^{EXPn}_4$ and related invariant reductions of the corresponding three dimensional PDE} \label{optimreductions}
 To find all reductions and in this way to find all classes of the nonequivalent group invariant solutions of a differential equation  Ovsiannikov \cite{Ovsiannikov1982} has introduced the idea of an optimal system of subalgebras for a given symmetry algebra of the differential equation. This idea is now widely used for PDEs and systems of ODEs arising in different areas of sciences \cite{Ovsiannikov1994}, \cite{Meleshko2005}, \cite{Meleshko2008}.

Now we will study a complete set of possible reductions of the three dimensional PDE (\ref{maingeneralExpNegative}) to two dimensional PDEs. For this purpose we need an optimal system of subalgebras of $L_4^{EXPn}$.
As before in \cite{BordagYamshchikov2017} we use an optimal system developed in \cite{paterawintern} for the real four dimensional Lie algebras of this type. To make the comparison of the results transparent we introduce in this Section the same notations for the generators of $L_4^{EXPn}$ as in \cite{paterawintern} and in \cite{BordagYamshchikov2017}.

In the basis (\ref{symalgenaralExpNegativ}) of $L^{EXPn}_4$ there are only two non-zero commutation relations (\ref{nonzeroComEXPn}). If we introduce notations like in the paper \cite{paterawintern}, i.e. we denote $\mathbf U_i=e_i$ where $i=1, \dots,4$ then we can rewrite the relations (\ref{nonzeroComEXPn}) as
\begin{equation} \nonumber
	[e_1,e_2] = e_2,~~~  [e_3, e_4]= e_4.
\end{equation}
Now we can see that $L^{EXPn}_4$ corresponds to the algebras of the type $A_2 \oplus A_2$  in the classification  of \cite{paterawintern} where also optimal systems of subalgebras for all real three and four dimensional solvable Lie algebras are provided. The corresponding system of optimal subalgebras of $L^{EXPn}_4$ is listed in Table \ref{TableOptSystL4EXPn}.

\begin{table}[h]
\begin{center}
\caption{ The optimal system of one-, two- and three- dimensional
subalgebras of $L^{EXPn}_4$,  where $\omega$ is a parameter such that $-\infty < \omega < \infty$. \label{TableOptSystL4EXPn}}
\begin{tabular}{|l|l|}
         \hline
         Dimension of & $~~~~~~~~~~~~~~~~~~~~~$System of optimal subalgebras of algebra $L^{EXPn}_4$ \\
         the subalgebra & \\
         \hline

         1                       & $h_1=\left< e_2 \right>, ~ h_2= \left< e_3 \right>, ~h_3=\left< e_4 \right>, ~h_4=\left< ~e_1 + \omega e_3 \right>, ~h_5=\left< ~e_1 \pm e_4 \right>,$ \\
         			& $h_6=\left< ~e_2 \pm e_4 \right>, ~h_7=\left< ~e_2 \pm e_3 \right> $\\
  \hline
         2                       & $h_ 8=\left < e_1, e_3 \right>, ~h_9=\left <e_1, e_4\right>,~ h_{10}=\left <e_2,e_3\right>,~ h_{11}=\left <e_2,e_4\right>,$\\
         			& $~ h_{12}=\left <e_1+\omega e_3,e_2\right>, ~ h_{13}=\left <e_3 + \omega e_1,e_4\right>, ~ h_{14}=\left <e_1 \pm e_4,e_2\right>,$\\
				& $~ h_{15}=\left <e_3 \pm e_2,e_4\right>, ~ h_{16}=\left <e_1 + e_3,e_2 \pm e_4\right>$\\

         \hline
          3                      & $h_{17}=\left < e_1, e_3, e_2 \right>, h_{18}=\left < e_1, e_4, e_2 \right>, h_{19}=\left < e_1, e_3, e_4 \right>, h_{20}=\left < e_2, e_3, e_4 \right>, $\\
				& $h_{21}=\left < e_1 \pm e_3, e_2, e_4 \right>, h_{22}=\left < e_1 + \omega e_3, e_2, e_4 \right>$\\
         \hline\noalign{\smallskip}
\end{tabular}
\end{center}
\end{table}

Now we are going to study all possible invariant reductions of the main equation (\ref{maingeneralExpNegative}).\\
 Let us first note that the subgroups $H_1$, $H_3$ and $H_6$ generated by subalgebras $h_1 = \left< \frac{\partial }{\partial V} \right>$, $h_3 = \left< e^{rt} \frac{\partial }{\partial l} \right>$ and $h_6 = \left< \frac{\partial }{\partial V} \pm  e^{rt} \frac{\partial }{\partial l} \right>$ correspondingly, do not give us any interesting reductions so we omit the detailed study of these cases here. We start with a first interesting and nontrivial case.

\textbf{Case} $H_2(h_2)$.  The subalgebra $h_2$ is spanned by the generator $e_3$
\begin{equation}
h_2 = < e_3 >= \left< -\frac{1}{a r}\left(e^{r t}\int{e^{-r t}d{\ln{\overline{\Phi}(t)}}}  \right)\frac{\partial}{\partial l} +  \frac{1}{r}\frac{\partial}{\partial t} \right>. \nonumber
\end{equation}
To find all invariants of the subgroup $H_2$ we solve the related characteristic system of equations
\begin{equation}
 \frac{dl}{ -\frac{1}{a r}\left(e^{r t}\int{e^{-r t}d{\ln{\overline{\Phi}(t)}}} \right) } = \frac{dt}{\frac{1}{r}} = \frac{dV}{0}=\frac{dh}{0}, \nonumber
\end{equation}
where the last two equations of the system present a formal notation that shows that the independent variable $h$ and the dependent variable $V$ are actually  invariants  under the action of the subgroup  $H_2$. We can obtain other independent invariants solving the system above. So we obtain  a set of independent invariants
\begin{eqnarray} \label{inv12H2EXPn}
inv_1 &=& z= l + \frac{1}{a r}e^{r t}\int{e^{-r t}d{\ln{\overline{\Phi}(t)}}} - \frac{1}{a r} \ln{\overline{\Phi}(t)}-\frac{1}{ar} (1+\ln a), ~~  inv_2 = {h}, ~~~~~~\\
inv_3 &=& W(z,h) =  V(l,h,t) . \label{inv3H2EXPn}
\end{eqnarray}
The invariants (\ref{inv12H2EXPn}) can be used as the new independent variables $z,h$ and the invariant (\ref{inv3H2EXPn}) as the new dependent variable $W(t,z)$ to reduce the three dimensional PDE (\ref{maingeneralExpNegative}) to a two dimensional one
\begin{eqnarray}\label{2dimEXPnH2}
&&\frac{1}{2} \eta^2 h^2 W_{hh}   + (\mu - \delta) h W_h +(r z + \delta h) W_z +\frac{1}{a} w_z \ln W_z  \\
&-& \frac{(\alpha - r)^2 W_z^2 + 2 (\alpha - r) \eta \rho h W_z  W_{zh}+ \eta^2 \rho^2 \sigma^2 h^2  W_{zh}^2}{2 \sigma^2 W_{zz} } =0. ~~~~~~\nonumber
\end{eqnarray}
In (\ref{cond:main}) we formulate the boundary condition and in Lemma \ref{PVF} we formulate the main  properties of the value function. Now we have to reformulate the boundary condition on the function $W(z,h)$ after the substitution (\ref{inv12H2EXPn}). To make further remarks transparent we take first as an example the simples form of the liquidation time distribution and suppose that $\overline{\Phi}(t)=e^{-\kappa t}$, i.e. we have to do with exponential liquidation time distribution. Then the new variable $z$ will take the form
\begin{equation}\label{z_exp_H2EXPn}
z= l + \frac{\kappa}{a r} t  +\frac{1}{ar} \left(\frac{\kappa}{r}-1-\ln a\right).
\end{equation}
It means that $z$ is increasing if $l$ or $t$ are growing up. But it leads to contradiction between the properties of the function $ W(z,h) =  V(l,h,t)$. One sides the boundary condition demands that the value function tends to zero for $t \to \infty$, other sides that the same function is strictly increasing by $l \to \infty$. Because  after the invariant substitution the new variable $z$ is the sum of these two old variables $l$ and $t$ we are not able to solve this contradiction. The similar inconsistency problem arising if we use another form of the function  $\overline{\Phi}(t)$.
Following this reduction cannot be used to solve  the optimization problem.

\textbf{Case} $H_4(h_4)$. Now we look for invariants of the subgroup $H_4$. The corresponding subalgebra $h_4$ is spanned by the generator $e_1 + \omega e_3$, i.e.
\begin{equation}
h_4 = \left<   \frac{1}{ar} \left(1 -\omega e^{r t}\int{e^{-r t}d{\ln{\overline{\Phi}(t)}}}  \right)\frac{\partial}{\partial l} +  \frac{\omega}{r}\frac{\partial}{\partial t} -  V \frac{\partial}{\partial V}    \right>. \nonumber
\end{equation}
We need to regard two special cases $\omega = 0$ and $\omega \neq 0$ here. If $\omega = 0$ then
\begin{equation}
h_4=<e_1 >=\left<\frac{1}{a r} \frac{\partial}{\partial l} - V \frac{\partial}{\partial V}\right>. \nonumber
\end{equation}
The invariants of the group $H_4$ are
\begin{eqnarray}\label{invH4EXPnspec}
inv_1 = h , ~~~ inv_2 = t ,~~~
inv_3 = W(h,t) = V(l,h,t) e^{a r l} \nonumber
\end{eqnarray}
From the last relation follows that $V(l,h,t)=W(h,t)e^{- a r l}$. We see that in this case the value function has the form $V(l,h,t) = e^{-a r l} W(h,t)$ and the complete dependence on $l$ is described just by the factor $e^{- a r l}$. It means that we obtain a decreasing function $V(l,h,t)$ in the variable $l$ in contradiction to the properties of a value function (see Lemma \ref{PVF}). It means that this reduction do not provide any meaningful solutions for our problem.

Now we can move according to a standard procedure to find the invariants of $H_4$  when $\omega \neq 0$. We obtain three independent invariants using a corresponding characteristic system
\begin{eqnarray}
\label{inv12H4EXPn}
inv_1 &=&z= l + \frac{1}{a r}e^{r t}\int{e^{-r t}d{\ln{\overline{\Phi}(t)}}} - \frac{1}{a r} \ln{\overline{\Phi}(t)}-\frac{t}{a\omega}, ~~  inv_2 = {h}, ~~~~~~\\
inv_3 &=& W(z,h) =  V(l,h,t) e^{\frac{r}{\omega}t}. \label{inv3H4EXPn}
\end{eqnarray}

Analogously substituting expressions for the invariants $z$  as the new independent variable and $W(z, h)$ as the new dependent variables into (\ref{maingeneralExpNegative}) we get
\begin{eqnarray}
&&\frac{1}{2} \eta^2 h^2 W_{hh}   + (\mu - \delta) h W_h +(r z + \delta h) W_z +\frac{1}{a} w_z \ln W_z \nonumber \\
&-& \frac{(\alpha - r)^2 W_z^2 + 2 (\alpha - r) \eta \rho h W_z  W_{zh}+ \eta^2 \rho^2 \sigma^2 h^2  W_{zh}^2}{2 \sigma^2 W_{zz} } \label{2dimEXPnH4} \\
&-&\frac{1}{a}\left(\frac{1}{ \omega}+(1+\ln a) \right) W_z  -\frac{r}{\omega}W=0. ~~~~~~\nonumber
\end{eqnarray}
We prove now a compatibility of the invariant substitutions (\ref{inv12H4EXPn}) and the boundary condition (\ref{cond:main}).
 As before we look for  the new invariant variables (\ref{inv12H4EXPn})-(\ref{inv3H4EXPn}) in the case of exponential liquidation time with $\overline{\Phi}(t)=e^{-\kappa t}$ then these formulas take the form
\begin{eqnarray}
 \label{exp_inv12H4EXPn}
z &=& l  + \frac{\kappa \omega -r}{a r \omega }t+ \frac{\kappa}{a r^2}, ~~~~~  inv_2 = {h}, ~~~~~~\\
 V(t,l,h)&=& W(z,h) e^{-\frac{r}{\omega}t}, ~~~\omega\ne 0. \label{exp_inv3H4EXPn}
\end{eqnarray}
From (\ref{exp_inv12H4EXPn}) follows that if we chose an arbitrary parameter $\omega = r/{\kappa}$ then the variable $z$ up to a constant shift coincides with the old variable $l$. The relation (\ref{exp_inv3H4EXPn}) shows that the boundary condition (\ref{cond:main}) will be satisfied for any solution of (\ref{2dimEXPnH4}). We see also that for other positive values of the parameter $\omega$ the invariant variables (\ref{exp_inv12H4EXPn})-(\ref{exp_inv3H4EXPn}) are compatible with the boundary condition (\ref{cond:main}).

Similar to the case of the exponential time distribution we can study other types of liquidation time distributions. For instance, we look at the frequently  used Weibull distribution with $\overline{\Phi}(t)=e^{-( t/{\lambda})^k}$ where the invariant variables will take the form
\begin{eqnarray}
 \label{Weibull_inv12H4EXPn}
z &=& l  + \frac{k }{a r^{k+1}\lambda^k}e^{r t}\Gamma(k,r t)+ \frac{1}{a r {\lambda}^k} t^k - \frac{1}{a \omega} t, ~~~~~  inv_2 = {h}, ~~~~~~\\
 V(t,l,h)&=& W(z,h) e^{-\frac{r}{\omega}t}, ~~~\omega\ne 0, \label{Weibull_inv3H4EXPn}
\end{eqnarray}
here $\Gamma(k,r t)$ is the upper incomplete gamma function.

For the studied optimization problem the most interesting case appears if the liquidation time distribution has a lokal maximum like we expect it in the real world. The Weibull distribution has a local maximum for the parameter $k>1$. Because of the asymptotic behavior of the expression $e^{r t}\Gamma(k,r t) \to r^{k-1}t^{k-1}$ as $t \to \infty$ we obtain that for $k>1$  the variable  $z\to l+\frac{1}{a r {\lambda}^k} t^k $ as $t \to \infty$. It means that also for a Weibull distribution we have compatibility of the invariant substitutions (\ref{Weibull_inv12H4EXPn})-(\ref{Weibull_inv3H4EXPn}) with the boundary condition (\ref{cond:main}).

We notice that the investment $\pi(z, h)$ and consumption $c(z, t, h)$ in the case $H_4$ look as
\begin{eqnarray}\label{reducedStragy}
	 \pi (z,h) =  \left(  - \frac{\eta \rho \sigma h W_{zh}+(\alpha - r)W_z}{\sigma^2 W_{zz} } \right),~
		 c(z,t, h) = \frac{1}{a} \ln{\left(\frac{\overline{\Phi}(t)}{a W_z }\right)}+\frac{r }{a \omega}t, ~~ \omega > 0, \nonumber
\end{eqnarray}
where  $W(z,h)$ is a solution of the equation (\ref{2dimEXPnH4}).

\textbf{Case} $H_5(h_5)$. According to the first line of  Table \ref{TableOptSystL4EXPn} the subalgebra $h_5$ corresponding to the subgroup $H_5$ algebra is spanned by
\begin{equation}
h_5 = <e_1 \pm e_4> = \left< \left( \frac{1}{a r} \pm e^{rt}  \right)\frac{\partial}{\partial l}  -  V \frac{\partial}{\partial V} \right>. \nonumber
\end{equation}
 Using a standard procedure to determine the invariants of the subgroup $H_5$ we obtain three independent invariants as a solution of the characteristic system, they have a form
\begin{eqnarray}\label{invH5LOG}
inv_1 = h, ~~~ inv_2 = t,~~
inv_3 = v(h,t) = e^{\frac{ar }{1\pm a r e^{r t} }l } V(l,h,t) .
\end{eqnarray}
It means also that the complete dependence of the value function $V(l,h,t)$ on the variable $l$ is described just by the factor $e^{-\frac{ar }{1\pm a r e^{r t} }l } $. If $t> \frac{1}{r}\ln (ar)$ then the value function will be  decreasing function in $l$ by choosing the plus sign in the denominator of the fraction $-\frac{ar }{1\pm a r e^{r t} }l$ and it will be increasing function in $l$ if we choose the  minus sign in the denominator of the fraction.

Because the value function for the optimization problem should be increasing function in $l$  so we need to study just this one case. Therefore we choose as a new dependent variable the function $v(h,t) = e^{\frac{ar }{1- a r e^{r t} }l } V(l,h,t)$.
Substituting the new dependent variable $v(h,t)$  into (\ref{maingeneralExpNegative}) we get a two dimensional PDE
\begin{eqnarray}\label{2dimEXPnH5}
&v_t +\frac{1}{2}\eta^2 h^2 v_{hh} +(\mu -\delta)hv_h +\frac{r }{ a r e^{r t} -1}\left(\left(a \delta h -1+\ln{\left( \frac{r \overline{\Phi}(t)}{ a r e^{r t} -1}      \right)}\right) v+v \ln v\right)\nonumber \\
&-\frac{(\alpha -r)^2}{2 \sigma^2 }v -\frac{(\alpha -r)\eta \rho h}{\sigma^2}v_h -\frac{(a r e^{r t} -1)^2\eta^2 \rho^2 h^2 }{2 a^2 r^2 }\frac{v_{h}^2}{v}=0, ~~~v(h,t)_{t \to \infty}{\longrightarrow} 0.\nonumber
\end{eqnarray}
 After Lemma \ref{PVF} the value function $V(l,h,t)$ cannot have an exponential growth in $l$ like we obtain it now. It means that the invariant substitution (\ref{invH5LOG}) is inconsistent with the possed optimization problem.

\textbf{Case} $H_7(h_7)$. The last one dimensional subalgebra in the list of the optimal system of subalgebras in Table \ref{TableOptSystL4EXPn} is spanned by $e_2 \pm e_3 $

\begin{equation}
h_7 = <e_2 \pm e_3 > = \left<\pm \left( -\frac{1}{a r}\left(e^{r t}\int{e^{-r t}d{\ln{\overline{\Phi}(t)}}}  \right)\frac{\partial}{\partial l} +  \frac{1}{r}\frac{\partial}{\partial t}       \right)+ \frac{\partial}{\partial V} \right>.\nonumber
 \end{equation}

 According to a standard procedure, we obtain following invariants of the subgroup $H_7$
\begin{eqnarray}\label{inv12H7EXPn}
inv_1 &=& z=l + \frac{1}{a r}e^{r t}\int{e^{-r t}d{\ln{\overline{\Phi}(t)}}} - \frac{1}{a r} \ln{\overline{\Phi}(t)}, ~~  inv_2 = {h}, ~~~~~~\\
inv_3 &=& W(z,h) =  V(t,l,h) \mp r t. \label{inv3H7EXPn}
\end{eqnarray}
Using these invariants (\ref{inv12H7EXPn}), (\ref{inv3H7EXPn}) as the new  variables $z,h, W(z,h)$  and substituting them into (\ref{maingeneralExpNegative}) we obtain a two dimensional PDE on $W(z,h)$
\begin{eqnarray}\label{2dimEXPnH7}
&&\frac{1}{2} \eta^2 h^2 W_{hh}   + (\mu - \delta) h W_h +(r z + \delta h) W_z +\frac{1}{a} w_z \ln W_z  \\
&-& \frac{(\alpha - r)^2 W_z^2 + 2 (\alpha - r) \eta \rho h W_z  W_{zh}+ \eta^2 \rho^2 \sigma^2 h^2  W_{zh}^2}{2 \sigma^2 W_{zz} } + \frac{1}{a}(1+\ln a)W_z \pm r=0 ~~~~~~\nonumber
\end{eqnarray}
 In this case we see the inconsistence between the boundary condition (\ref{cond:main}) which demands that $ V(t,l,h) \to 0$ as $t\to \infty$ and the invariant substitutions (\ref{inv12H7EXPn}), (\ref{inv3H7EXPn}) which say that the expression $ V(t,l,h) \mp r t$ depends just on $z,h$ and not from the variable $t$.

  \vspace{5pt}
Totally there are four meaningful reductions of the three dimensional PDE (\ref{maingeneralExpNegative}) for the case of the EXPn utility function and the general liquidation time distribution $\overline{\Phi}(t)$ by using one dimensional subalgebras of the algebra $L_4^{EXPn}$. Just one of these reductions which corresponds to the case $H_4$ with $\omega \ne 0$, i.e.  the substitutions (\ref{inv12H4EXPn}), (\ref{inv3H4EXPn}) are consistent with the boundary condition (\ref{cond:main}) and the two dimensional PDE (\ref{2dimEXPnH4}) is a corresponding reduction. This equation can be studied further with numerical methods.
  \vspace{5pt}

 In Table~\ref{TableOptSystL4EXPn} are listed also two and three dimensional subalgebras of $L^{EXPn}_4$. Using these subalgebras maybe we can find the deeper reductions of the PDE (\ref{maingeneralExpNegative}) for instance  to  ordinary differential equations.

\textbf{Case} $H_8(h_8)$. We take the first two dimensional subalgebra listed in Table~\ref{TableOptSystL4EXPn}, i.e. the  subalgebra $h_8 = < e_1, e_3 >$. We rewrite the characteristic systems to the first generator $e_1$  in terms of the invariants of $e_3$ (\ref{inv12H2EXPn}), (\ref{inv3H2EXPn}) then $e_1$ takes the form  $e_1=\frac{1}{a r }\frac{\partial}{\partial z} - W \frac{\partial}{\partial W}.$ Solving a corresponding characteristic system we obtain a new invariant
 \begin{equation}
inv_{e_1} = v(h) =   W(z,h)e^{ar z} ,\nonumber
\end{equation}
 which we use now as a new dependent variable to reduce the equation (\ref{2dimEXPnH2}) to an ODE
 \begin{eqnarray}\label{ODEh8EXPn}
&\frac{1}{2}\eta^2 h^2 v^{''}-\frac{\eta^2 \rho^2  }{2 }h^2 \frac{\left(v^{'}\right)^2}{v}
+\left(\frac{(\mu -\delta)\sigma^2-(\alpha -r)\eta \rho }{\sigma^2}\right)h v^{'} \nonumber\\
& - \left(a r  \delta h
+\frac{(\alpha -r)^2}{2 \sigma^2 }\right)v
-r v \ln{(-a r v)}=0.
\end{eqnarray}
In terms of original variables $l, h, t$ and $V(l,h,t)$ the substitution looks as follows
\begin{eqnarray} \label{subh8EXPn}
 V (l,h,t) &=& v(h)e^{- ar z},\\
z &=&  l + \frac{1}{a r}e^{r t}\int{e^{-r t}d{\ln{\overline{\Phi}(t)}}} - \frac{1}{a r} \ln{\overline{\Phi}(t)}-\frac{1}{ar} (1+\ln a).\nonumber
\end{eqnarray}
Now we obtain a reduction of the three dimensional PDE (\ref{maingeneralExpNegative}) to an ODE.
 But we cannot use this reduction, because it is inconsistent with the properties of the value function $V(l,h,t)$ listed in the Lemma~\ref{PVF}. The value function is an increasing function in variable $l$ and $V(l,h,t)>0$, it means also $v(h)$ should be a positive function. From the first expression in (\ref{subh8EXPn}) follows that $V(l,h,t)$ is decreasing in $z$ and following in the variable  $l$ and from the equation (\ref{ODEh8EXPn}) follows that the expression $\ln{(-a r v)}$ is well defined just for negative functions $v(h)$.

\textbf{Case} $H_{12}(h_{12})$. Similar to the previous case we study now the case of  $h_{12} = < (e_1+\omega e_3), e_2 >$ and after the substitution  \begin{eqnarray}
\label{inv12H12EXPn}
  V(t,l,h) &=& v(h)\exp{\left(-a r l-  e^{r t}\int{e^{-r t}d{\ln{\overline{\Phi}(t)}}} +  \ln{\overline{\Phi}(t)}\right)}
\end{eqnarray}
we obtain an ODE on the function $v(h)$
\begin{eqnarray}\label{ODEh12EXPn}
&\frac{1}{2}\eta^2 h^2 v^{''}-\frac{\eta^2 \rho^2  }{2 }h^2 \frac{\left(v^{'}\right)^2}{v}
+\left(\frac{(\mu -\delta)\sigma^2-(\alpha -r)\eta \rho }{\sigma^2}\right)h v^{'} \nonumber\\
& - \left(a r  \delta h
+\frac{(\alpha -r)^2}{2 \sigma^2 }\right)v +r(1+\ln a)v
-r v \ln{(-a r v)}=0.
\end{eqnarray}
If we can find a positive solution to this equation then we get the solution to the original optimization problem. It is easy to see that the last term in the equation (\ref{ODEh12EXPn}) will be complex-valued if the function $v(h)>0$.  It means that it is not possible to find a positive solution of this equation. This reduction is not compatible with the conditions possed on the optimization problem. Like in the previous case we see that also other properties of the value function listed in the Lemma~\ref{PVF}  cannot be satisfied if the value function takes the form (\ref {inv12H12EXPn}).

All other two -  and three -  dimensional subalgebras listed in Table~\ref{TableOptSystL4EXPn} do not give any meaningful  reductions of the original equation (\ref{maingeneralExpNegative}), so we will not regard them in detail.

We studied the complete set of all possible reductions of the original three dimensional PDE (\ref{maingeneralExpNegative}) to simpler differential equations. We see that not all of the reductions are reasonable for the optimization problem. Just one of them represented by the two dimensional PDE  satisfies all conditions.  It is the main result of this Section and we formulate this result as a theorem
 \begin{theorem}\label{reductiontheorem5}
The main three dimensional HJB equation (\ref{maingeneralExpNegative}) admits the unique symmetry  reduction to the two dimensional PDE (\ref{2dimEXPnH4}) after the substitutions (\ref{exp_inv12H4EXPn}) - (\ref{exp_inv3H4EXPn}) which satisfies all conditions of the posed optimization problem. The corresponding investment - consumption strategies are given in (\ref{reducedStragy}).
\end{theorem}

\section{Conclusion}\label{conclusion}

In this paper we study a portfolio optimization problem for a basket consisting of a risk free liquid, risky liquid and risky illiquid assets where the investor prefer to use an exponential utility function. The illiquid asset is sold in a random moment $T$ with a known distribution of the liquidation time. It is a distribution with a survival function $\overline{\Phi} (t)$, satisfying very general conditions  $\lim_{t \to \infty} \overline{\Phi} (t) E[U(c(t))]=0$ and $\overline{\Phi} (t) \sim e^{-\kappa t}$ or faster as $t \to \infty$. Typically one suppose that the liquidation time distribution is an exponential one, i.e. $\overline{\Phi} (t) = e^{-\kappa t}$, $ t\ge 0, ~\kappa>0, ~$ or of the Weibull type with $\overline{\Phi} (t)=e^{-(t/\lambda)^k}$, with  $t\ge 0,~ k>0,~ \lambda >0$. The Weibull distribution turns to the exponential distribution by $k=1$ and it can be understand as  a generalization of the exponential distribution. Based on the economical motivation we choose $k>1$ because of in this case the Weibull  probability density function has a local maximum.

Before in  papers   \cite{BordagYamshchikov2017}, \cite{Bordag2016}  we studied   similar portfolio optimization problems where the investor used the HARA  and LOG  utility functions correspondingly instead the exponential utility function as in this paper.

Both  the HARA  utility function  as well as the LOG  utility function were widely used before in  optimization problems with a random income and for different settings of the portfolio optimization problems. Usually it was going on the optimization  problems with a portfolio that includes an illiquid asset that was sold in a deterministic moment of time, i.e. on a portfolio optimization problem with a finite time horizon. Other authors  supposed that the illiquid asset is not sold  at all, i.e. they studied a portfolio optimization problem with the infinite time horizon. In previous papers \cite{BordagYamshchikov2017}, \cite{Bordag2016} we demonstrated the connection between these  two problems. We also shown that for $\gamma \to 0$  we obtain ${U_1^{HARA}}_{\gamma \to 0}{\longrightarrow} U^{LOG}$ as well as  a three dimensional HJB equation (\ref{eq:HJB21}) corresponding to the HARA utility function formally transforms into the HJB equation with the LOG utility function. Then we proved independently from the form of the survival function $\overline{\Phi} (t)$ that the Lie algebraic structure of the PDE with Lie logarithmic utility can be seen as a limit of the algebraic structure of the PDE with the HARA utility function as $\gamma \to 0$.

Now we provided a complete Lie group symmetry analysis for two different exponential utility functions, EXPn and EXPp, i.e. for two three dimensional PDEs (\ref{maingeneralExpNegative}) and (\ref{maingeneralExpPositive}) which contain an arbitrary function $\overline{\Phi} (t)$. The results  are formulated in Theorem~\ref{THEXPn} and Theorem~\ref{TEXPp}.
We obtained that each of these PDEs admitted the four dimensional Lie algebras $L_4^{EXPn}$ and $L_4^{EXPp}$ correspondingly. These algebras are isomorph and similar, it means that the studied PDEs (\ref{maingeneralExpNegative}) and (\ref{maingeneralExpPositive}) are equivalent up to the one-to-one analytical substitution. In other words the optimization problems are identical from any point of view: an economical, analytical or Lie algebraic one.

We also investigated a connection between the optimization problem with the HARA utility function (\ref{harautility2}) and with the EXPn utility function in Section~\ref{HARAtoEXPn}. Even though the HARA utility function is connected to the negative exponential utility function by $\gamma \to \infty$ as we mentioned in (\ref{expnegative}) we do not get the expected connection between the corresponding optimization problems. Instead of that we obtain quite different structures of the invariant variables by the study of the symmetry reductions of the main equation (\ref{maingeneralExpNegative}). In the case of the HARA utility function a typical invariant variables were the fraction $\frac{l}{h}$ and time $t$. It means that in the case of the HARA or LOG utility function the value function depends in the first place from the relation between the values of the liquid  and illiquid assets. It is completely independent of the absolute value of his liquid part or from the absolute value of his illiquid part of wealth. For instance, the investor in the HARA  case, as we proved it before, should increase his consumption rapidly if the relation l/h falls, independently how many millions of dollars the investor has as liquid part at the moment.

Here in the case of the exponential utility function the situation is quite different.  As follows from equation (\ref{inv12H4EXPn}), the behavior of the investor depends now on two variables, on the value of the illiquid asset $h$ and on the combined variable $z$ which contains the liquid part of wealth and an economically modulated time. As a consequence, the absolute value of the illiquid part of wealth plays a large role. The variable $z$ tells us that the influence of a large amount of a liquid asset plays the same role as the possibility to wait a long time.
By the way,  this difference in the behavior of the invariant variables and the radical change of the investment – consumption strategies  is to explain  by the fact that the risk tolerance in the case of the HARA utility function is a linear function of the consumption $c$ and in the case of the exponential utility the risk tolerance is just a constant $R(c)=\frac{1}{a}$.

A further difference between the optimization problems with the HARA  and EXPn  utility functions is related to the structure of the admitted Lie algebras. In the cases of the HARA  and LOG  utility functions the corresponding three dimensional PDEs have admitted   three dimensional main Lie algebras. Just by the special choice of a liquidation time distribution,  i.e. only for the exponential function $\overline{\Phi} (t) = e^{-\kappa t}$ we got an extension of these Lie algebras to the four dimensional ones. Here in the cases of the exponential utility functions, independently EXPn or EXPp, we obtain from the beginning the four dimensional Lie algebras as the symmetry algebras of the corresponding PDEs. It is remarkable that in these cases the four dimensional Lie algebras do not allow any extension independently from the form of $\overline{\Phi} (t)$. It can be seen by the  solving of the  system of equation (\ref{lieSystemExpNegativ}) in the proof of Theorem~\ref{THEXPn}.

 In the previous paper \cite{BordagYamshchikov2017} we proved that the algebra $L^{LOG}_4$ can be obtained as a limit case of $L^{HARA}_4$ by $\gamma \to 0$. Here we see that $L^{HARA_2}_3$  (or correspondingly $L^{HARA_2}_4$) and $L_4^{EXPn}$ are quite different and they do not connect by $\gamma \to \infty$ as well as they do not
have any connections between analytical structures of their generators independently on the form of the liquidation time distribution.

In our paper we pay attention to the internal structure of the admitted algebra $L_4^{EXPn}$  to obtain convenient and useful reductions of the main equation (\ref{maingeneralExpNegative}). Further on we use the system of optimal subalgebras provided in \cite{paterawintern} and get corresponding nonequivalent invariant reductions of the three dimensional PDEs (\ref{maingeneralExpNegative}) to  two dimensional PDEs.  They describe the complete set of solutions that can not be transformed into each other with  the help of the transformations of the admitted symmetry group. We show that the three dimensional PDE can be reduced to a corresponding two dimensional ones in Section~\ref{optimreductions}. The low dimensional PDEs are much more convenient for further analytical or numerical studies. We also provide the formulas for the optimal investment-consumption policies in invariant variables using solutions of the reduced equation. We demonstrate that between meaningful reductions there exists one  (\ref{2dimEXPnH4}) which is consistent with the boundary condition (\ref{cond:main}) and with the expected properties of the value function.

We remark also a different level of influence of the  parameters on the HJB equation and on the admitted  Lie algebraic structure. The HJB equation contains seven parameters $ r,\alpha, \sigma$, $\mu, \delta, \eta, \rho$ which define the behavior of the  liquid and illiquid asset, and one parameter $a$ which is fixed by the exponential utility function. There are also some parameters which define the liquidation time distribution, for instance, it is the parameter $\kappa$ if we take the exponential distribution with $\overline{\Phi} (t) = e^{-\kappa t}$ or two parameters $\lambda$ and $k$ if we take the Weibull distribution with $\overline{\Phi} (t)=e^{-(t/\lambda)^k}$. If we look at the structure of the Lie algebras provided in  Theorem~\ref{THEXPn} and Theorem~\ref{TEXPp} we see that the generators of the  algebras are defined by the parameters $r, a$ and parameters of the liquidation time distribution only. The algebras change their structure if one or some of these parameters vanishing. Roughly said the most influence on the form of the solution of this optimization problem has interest rate $r$, the type of the investor's utility function and a  marked defined liquidation time distribution for the illiquid asset. These parameters define the invariant variables and the analytical structure of the solutions.

 Summing up, we carry a complete Lie group analysis for the optimization problems with negative and positive exponential utility functions and for a general  liquidation time distribution. We determine the reduced equation and corresponding optimal policies as it formulates in the Theorem~\ref{reductiontheorem5}.

\begin{acknowledgements}
The author is thankful to prof. L. Vostrikova-Jacod for interesting discussions and for organizing a very successful conference {\sf Advanced Methods in Mathematical Finance, Angers, 2018, France}, where the author got the idea to write this paper.
\end{acknowledgements}

%

\end{document}